\documentclass[
  twocolumn,
  aps,
  prl,
  showpacs,
  superscriptaddress,
  preprintnumbers,
  letterpaper
]{revtex4}

\usepackage{amsthm}
\usepackage{amsmath}
\usepackage{amssymb}
\usepackage{bbold}
\usepackage{graphicx}
\usepackage{subfigure}

\usepackage{tikz}
\usepgflibrary{shapes.geometric}
\usetikzlibrary{arrows,decorations.pathreplacing}

\newtheorem{theorem}{Theorem}
\newtheorem{lemma}{Lemma}

\theoremstyle{definition}

\newcommand*{\cH}{\mathcal{H}}

\newcommand*{\cL}{\mathcal{L}}
\newcommand*{\cM}{\mathcal{M}}

\newcommand*{\cR}{\mathcal{R}}

\newcommand{\ket}[1]{|#1\rangle}
\newcommand{\bra}[1]{\langle #1 |}

\newcommand{\braket}[2]{\langle #1 |#2 \rangle}

\newcommand{\proj}[1]{\ket{#1}\bra{#1}}

\newcommand{\beq}{\begin{equation}}
\newcommand{\eeq}{\end{equation}}
\newcommand{\best}{\begin{equation*}}
\newcommand{\eest}{\end{equation*}}

\newcommand{\idmap}{{\rm id}}

\DeclareMathOperator{\Tr}{Tr}

\begin{document}

\title{Hierarchy of efficiently computable and faithful lower bounds to quantum discord}

\author{Marco Piani}
\affiliation{SUPA and Department of Physics, University of Strathclyde, Glasgow G4 0NG, UK}
\affiliation{Department of Physics \& Astronomy and Institute for Quantum Computing, \\  University of Waterloo, Waterloo, Ontario, N2L 3G1, Canada}

\begin{abstract}
Quantum discord expresses a fundamental non-classicality of correlations more general than quantum entanglement. We combine the no-local-broadcasting theorem, semidefinite-programming characterizations of quantum fidelity and quantum separability, and a recent breakthrough result of Fawzi and Renner about quantum Markov chains to provide a hierarchy of computationally efficient lower bounds to quantum discord. Such a hierarchy converges to the surprisal of measurement recoverability introduced by Seshadreesan and Wilde, and provides a faithful lower bound to quantum discord already at the lowest non-trivial level. Furthermore, the latter constitutes by itself a valid discord-like measure of the quantumness of correlations.
\end{abstract}

\maketitle

\emph{Introduction.}---Correlations in quantum mechanics exhibit non-classical features that include non-locality~\cite{revnonloc}, steering~\cite{wisemanPRL2007}, entanglement~\cite{revent}, and quantum discord~\cite{modi2012classical}. Quantum correlations play a fundamental role in quantum information processing and quantum technologies~\cite{nielsen2010quantum}, which go from quantum cryptography~\cite{QKD} to quantum metrology~\cite{giovannetti2011advances}. While both non-locality and steering are manifestations of entanglement, 
quantum discord is a more general form of quantumness of correlations that includes entanglement but goes beyond it.
In particular, \emph{almost all} distributed states exhibit discord~\cite{ferraro2010almost}.
This fact calls for fully elevating the study of quantum discord
to the quantitative level, since just certifying that discord is present may be considered of limited interest. While several approaches to the quantification of discord have been already proposed (see, e.g.~\cite{modi2012classical,luo2008using,modi2010unified,dakic2010necessary,luo2010geometric,streltsov2011linking,piani2011all,girolami2012observable,piani2012quantumness,paula2013geometric,chang2013remedying,girolami2014quantum,piani2014quantumness,seshadreesan2014r,seshadreesan2014fidelity} and references therein),
in this paper we significantly push forward a meaningful, reliable, and practical quantitative approach to the study of quantum discord that is based on fundamental quantum features of quantum correlations, and at the same time is computationally friendly.

Quantum discord was introduced in terms of the minimum amount of correlations, as quantified by mutual information, that is necessarily lost in a local quantum measurement of a bipartite quantum state~\cite{ollivier2001quantum,henderson2001classical} (see below for exact definitions). It is then clear that it is relatively easy to find upper bounds to quantum discord: the loss of correlations due to \emph{any} measurement
provides some upper bound. Nonetheless, standard quantum discord is not easily computed even in simple cases, and general easily computable \emph{lower} bounds to it are similarly not known. In this paper we provide a family of lower bounds for the standard quantum discord which can reliably be computed numerically. On the other hand, they have each physical meaning, since they are based on `impossibility features' (i.e., no-go theorems) 
related to the local manipulation of quantum correlations. Furthermore, such lower bounds satisfy the basic requests that should be imposed on any meaningful measure of quantum correlations~\cite{brodutch2012criteria,piani2012problem}, hence making each quantifier in the hierarchy a valid discord-like quantifier in itself.

One `impossibility feature' associated to quantum discord relates to local broadcasting~\cite{piani2008no,luo2010decomposition}, which can be seen as a generalization of broadcasting~\cite{barnum1996noncommuting}, itself a generalization of cloning~\cite{wootters1982single}: correlations that exhibit quantum discord cannot be freely locally redistributed or shared, and indeed, discord can be exactly interpreted as the asymptotic loss in correlations necessarily associated with such an attempt~\cite{streltsov2013quantum,brandao2013quantum}. A very related `impossibility feature' of discord deals with the `local relocation' of quantum states by classical means, that is, roughly speaking, with the transmission (equivalently, storing) of the quantum information contained in quantum subsystems via classical communication (a classical memory). Indeed, it can be checked through a powerful result by Petz~\cite{petz1986sufficient,petz1988sufficiency,hayden2004structure} that the ability to perfectly locally broadcast (equivalently, to perfectly store by classical means) distributed quantum \emph{states} reduces to the ability to perfectly locally broadcast or classically store \emph{correlations}, as measured by the quantum mutual information~\cite{hayashi2006quantum,piani2008no,luo2010decomposition}, a feat possible---by definition---only in absence of discord. The relation between the above two `impossibility features' is due to the fact that quantum information becomes classical when broadcast to many parties~\cite{bae2006asymptotic,chiribella2006quantum,chiribella2011quantum,brandao2013quantum}. 

The consideration of the general, non-perfect (for states exhibiting discord) case of the classical transmission/storing of an arbitrary quantum state has recently received renewed attention also thanks to a breakthrough result of Fawzi and Renner~\cite{fawzi2014quantum} (see also~\cite{brandao2014quantum}) that generalizes the result by Petz. In~\cite{seshadreesan2014fidelity},  Seshadreesan and Wilde explicitly suggested to approach the study of the general quantumness of correlations, and in particular their quantification, in terms of how well distributed quantum states can be locally transmitted or stored by classical means. 
They introduced a discord-like quantifier, the surprisal of measurement recoverability, which, thanks to the results of~\cite{fawzi2014quantum}, directly translates into a lower bound to the standard quantum discord. Unfortunately, the surprisal of measurement recoverability is in general not easily computable either. In this paper, by considering how well a quantum state can be locally broadcast, 
 we generalize the surprisal of measurement recoverability in such a way to obtain numerically computable (upper and) lower bounds to it, which provably converge to it. Thus, we also obtain computable lower bounds to the standard quantum discord.  
The hierarchy of lower bounds that we introduce exploits ideas used in the characterization and detection of entanglement via semidefinite programming~\cite{DohertyPRL,DohertyPRA,DohertyBell}. Semidefinite programming optimization techniques~\cite{boyd2009convex} have  found many other significant applications in quantum information (see, e.g.,~\cite{nowakowski,jain2011qip,kempe2010unique,navascuesPRL,watrous2009semidefinite,johnston2010family,eisertSDP}), and, in recent times, they have been used also in the quantification of steering~\cite{quantifyingsteering,piani2014einstein}. Here we extend the use of semidefinite programming for the study of quantum correlations to quantum discord.

\emph{Preliminaries}---We will consider finite-dimensional systems, so that a quantum state corresponds to a $d$-dimensional positive semidefinite density matrix $\rho$ which lives in the space $L(\cH)$ of linear operators on a Hilbert space $\cH\simeq\mathbb{C}^d$. The von Neumann entropy associated with $\rho$ is given by $S(\rho)=-\Tr(\rho\log\rho)$. We will indicate by $\Tr_{\backslash X}$ a trace performed over every other system except $X$. In the case we consider a bi- or multi-partite system, with global state $\rho$, we denote $S(X)_\rho=S(\rho_X)$, where $L(\cH_X)\ni\rho_X=\Tr_{\backslash X}(\rho)$ is the reduced state of system $X$.
The fidelity $F(\sigma,\rho)=\Tr\sqrt{\sqrt{\rho}\sigma\sqrt{\rho}}$  is a measure of how similar two states $\rho$ and $\sigma$ are~\cite{nielsen2010quantum}. It holds $0\leq F(\sigma,\rho)\leq 1$, with $F(\sigma,\rho)=1$ if and only if $\rho=\sigma$. We will need the fact that the fidelity can be seen as the solution to the semidefinite programming (SDP) optimization problem~\cite{killoran2012entanglement,watrous2012simpler}
\begin{subequations}
  \label{eq:SDPfidelity}
  \begin{align}
    {\text{maximize}}\quad
    & \frac{1}{2}(\Tr(X)+\Tr(X^\dagger)) \\
    \text{subject to}\quad
    &  \begin{pmatrix}
    	\rho 		& X \\
     	X^\dagger 	& \sigma
     \end{pmatrix} \geq 0.
   \end{align}
\end{subequations}
Another measure of similarity of states is the trace distance $\Delta(\sigma,\rho)=\frac{1}{2}\|\sigma -\rho\|_1$ where $\|\xi\|_1=\Tr(\sqrt{\xi^\dagger \xi})$ is the trace norm~\cite{nielsen2010quantum}. 
It holds $0\leq\Delta(\sigma,\rho)\leq 1$, and $1-F(\sigma,\rho) \leq \Delta(\sigma,\rho) \leq \sqrt{1-F^2(\sigma,\rho)}$~\cite{fuchs1999cryptographic}. Transformations of physical systems are described by completely positive and trace-preserving linear maps, also called channels, from $L(\cH_\textrm{in})$ to $L(\cH_\textrm{out})$, where $\cH_\textrm{in}$ and $\cH_\textrm{out}$ are the input and output spaces, respectively~\cite{nielsen2010quantum}.

\emph{Separability and symmetric extensions.}---A bipartite state $\rho_{AB}$ is separable (or unentangled) if it admits the decomposition 
$\rho_{AB}^\textrm{sep}=\sum_b p_b \proj{\alpha_b}_A\otimes \proj{\beta_b}_B$, 
for $\{p_b\}$ a probability distribution, and $\ket{\alpha_b}_A$ and $\ket{\beta_b}_B$ (not necessarily orthogonal) vector states for $A$ and $B$, respectively. A bipartite state that is not separable is entangled~\cite{werner1989}.

Consider systems $B_1\simeq  B_2 \simeq B$. A state $\rho_{AB_1B_2}$ such that $\rho_{AB_1}=\rho_{AB_2}=\rho_{AB}$, and such that $\rho_{AB_1B_2}=V_{B_1B_2}\rho_{AB_1B_2}V_{B_1B_2}^\dagger$, with $V_{B_1B_2}$ the swap operator, $V_{B_1B_2}\ket{\beta}_{B_1}\ket{\beta'}_{B_2}=\ket{\beta'}_{B_1}\ket{\beta}_{B_2}$, is called a (two-)\emph{symmetric extension} (on $B$) of $\rho_{AB}$. If the stronger condition $\rho_{AB_1B_2}=\Pi^+_{B_1B_2} \rho_{AB_1B_2}\Pi^+_{B_1B_2} $ holds, with $\Pi^+_{B_1B_2}$ the projector onto the symmetric subspace of $B_1B_2$, we call $\rho_{AB_1B_2}$ a (two-)\emph{Bose}-symmetric extension (on $B$). The concept can be generalized to $k$ extensions. Let $B^k=B_1B_2\ldots B_k$. We say that $\rho_{AB^k}$ is a $k$-symmetric extension of $\rho_{AB}$ (on $B$) if: (i) $\rho_{AB_i}=\Tr_{\backslash{AB_i}}(\rho_{AB^k})=\rho_{AB}$, for all $i=1,\dots,k$; (ii) $\rho_{AB^k}=V \rho_{AB^k} V^\dagger$ for any unitary $V$ that permutes the $B^k$ systems. Note that, because of the symmetry (ii), in (i) it is enough to consider the trace over all systems $B_i$ except an arbitrary one, e.g., $B_1$.
 If the stronger condition (ii') $\rho_{AB^k}=\Pi^+_{B^k} \rho_{AB^k} \Pi^+_{B^k}$, with $\Pi^+_{B^k}$ the projector onto the fully symmetric subspace $B^k_+$ of $B^k$, holds, we say that $\rho_{AB^k}$ is a $k$-Bose-symmetric extension of $\rho_{AB}$ (on $B$). Only separable states admit $k$-symmetric extensions for all $k$~\cite{fannes1988symmetric,raggio1989quantum,DohertyBell}.


\emph{No local broadcasting}.---The no-local-broadcasting 
theorem~\cite{piani2008no,luo2010decomposition} states that there exists a broadcasting channel $\Lambda_{B\rightarrow B_1B_2}$ such that
$\Tr_{B_1}(\Lambda_{B\rightarrow B_1B_2}[\rho^{AB}])=\Tr_{B_2}(\Lambda_{B\rightarrow B_1B_2}[\rho^{AB}])=\rho^{AB}$
if and only if $\rho^{AB}$ is quantum-classical, i.e., of the form
\beq
\label{eq:QCstate}
\rho_{AB}=\sum_b p_b \rho_b^A\otimes \proj{b}_B,
\eeq
with orthogonal $\ket{b}$s, and $\{p_b\}$ a probability distribution. Notice that here we focus on \emph{one-sided} local broadcasting~\cite{luo2010decomposition}, rather than two-sided local broadcasting~\cite{piani2008no}. 
If local broadcasting is possible, then: (i) it can be realized with a \emph{symmetric} broadcasting channel, whose output is symmetric among the outputs; (ii) an arbitrary number $k$ of extensions can be obtained, simply by $\ket{b}\mapsto\ket{b}^{\otimes k}$, for $\ket{b}$ as in~\eqref{eq:QCstate},  i.e., with output into the fully symmetric subspace $B^k_+$, so that the broadcasting channel has actually Bose-symmetric output (see Fig.~\ref{fig:broadcasting}).

Consider then Bose-symmetric broadcast maps $\Lambda_{B\rightarrow B^k_+}$ with output in the fully symmetric subspace $B^k_+$, and the induced maps $\Lambda^{\textrm{Sym}_+(k)}_{B} = \Tr_{\backslash B_1}\circ \Lambda_{B\rightarrow B^k_+}$, where $\circ$ denotes composition. We say that any map $\Lambda^{\textrm{Sym}_+(k)}_{B}$ that admits such a representation is $k$-Bose-symmetric extendible. The no-local-broadcasting theorem can then be recast as the fact that, for any $\rho_{AB}$ that is not quantum-classical, $F(\rho_{AB},\Lambda^{\textrm{Sym}_+(k)}_B[\rho_{AB}])< 1$ for any $k\geq 2$ and any $k$-Bose-symmetric extendible $\Lambda^{\textrm{Sym}_+(k)}_B$.

\begin{figure}
\begin{tikzpicture}[
       bluebox/.style={draw=black!50, minimum width=7mm, minimum height=30mm,
        fill=blue!10}]
        \node (rhoAB) at (0,0) {$\rho_{AB}$};
        \node (Lambda) at (2,1.5) [bluebox] {$\Lambda_{B\rightarrow B^k_+}$};
        \node (B1) at ([xshift=0.8cm,yshift=-1cm]Lambda.east) {$B_1$};
	  \node (B2) at ([xshift=0.8cm,yshift=-0.4cm]Lambda.east) {$B_2$};
	  \node (Bk) at ([xshift=0.8cm,yshift=1cm]Lambda.east) {$B_k$};
	  \node (A) at  ([xshift=+8mm,yshift=-2cm]Lambda.east) {$A$};
	  \node (B) at  ([xshift=-8mm]Lambda.west) {$B$};

        \draw[->] (rhoAB.east) -- ([xshift=-5mm]Lambda.west) -- (Lambda.west);
        \draw[->] (rhoAB.east) -- ([xshift=-5mm,yshift=-2cm]Lambda.west) -- (A.west);
        \draw[->] ([yshift=-1cm]Lambda.east) -- (B1.west) ;
   	 \draw[->] ([yshift=-0.4cm]Lambda.east) -- (B2.west) ;
	  \draw[->] ([yshift=+1cm]Lambda.east) -- (Bk.west) ;
	  
	   \node at ([xshift=0.25cm,yshift=0.4cm]Lambda.east) {$\vdots$};
	   
	  
	  
	  \draw [decorate,decoration={brace,amplitude=2pt,mirror}] ([xshift=1.2cm,yshift=-2.2cm]Lambda.east) --  ([xshift=1.2cm,yshift=-0.8cm]Lambda.east)  node [black,midway,xshift=1.2cm] {$\rho_{AB_1} \stackrel{\textrm{?}}{\approx} \rho_{AB}$};	  
\end{tikzpicture}
\caption{Symmetric local broadcasting (colour online). A local Bose-$k$-symmetric broadcasting channel $\Lambda_{B\rightarrow B^k_+}$ maps $B$ to the fully symmetric subspace of $B^k=B_1 B_2\ldots B_k$. The degree to which $\rho_{AB_1}=\Tr_{\backslash AB_1}(\Lambda_{B\rightarrow B^k_+}[\rho_{AB}])$ can approximate the initial state $\rho_{AB}$ depends on the classicality of correlations between $A$ and $B$.}
\label{fig:broadcasting}
\end{figure}
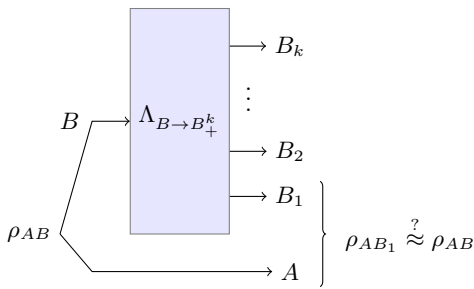

We now recall that every $k$-Bose-symmetric extendible channel is close to an entanglement-breaking (EB)---also called measure-and-prepare---map~\cite{horodecki2003entanglement}
\beq
\label{eq:EB}
\Lambda^{\text{EB}}_{B}[\cdot] = \sum_y \Tr(M^B_y \cdot) \proj{\beta_y}_B,
\eeq
where $\{M^B_y\}$ is a positive-operator-valued measure (POVM) and $\ket{\beta_y}_B$s are normalized vector states, not necessarily orthogonal. Entanglement-breaking maps have the defining property that, for any given $\rho_{AB}$, $(\idmap_A\otimes \Lambda^{\text{EB}}_{B}) [\rho_{AB}]$ is separable. 
%
%
More precisely, one can prove that for any $k$-Bose-symmetric extendible $\Lambda^{\textrm{Sym}_+(k)}_{B}$ there is an entanglement breaking map $\Lambda^{\text{EB}}_{B}$ close to it in the so-called diamond-norm distance; more precisely~\cite{chiribella2011quantum}:
\beq
\label{eq:chiribella}
\sup_{\rho_{AB}} \Delta\left(\Lambda^{\textrm{Sym}_+(k)}_{B}[\rho_{AB}],\Lambda^{\text{EB}}_{B}[\rho_{AB}]\right)\\
\leq \frac{|B|}{k},
\eeq
where $|B|$ indicates the dimension of system $B$, i.e., of $\cH_B$. Furthermore, it is clear that any entanglement-breaking map is $k$-Bose-symmetric extendible, since for any entanglement-breaking map \eqref{eq:EB} we can consider $\Lambda_{B\rightarrow B^k_+}[\cdot] = \sum_y \Tr(M^B_y \cdot) (\proj{\beta_y}^{\otimes k})_{B^k_+}$. Denote by $\cL^{\textrm{Sym}_+(k)}=\{\Lambda^{\textrm{Sym}_+(k)}\}$ the class of channels with a $k$-Bose-symmetric extension, and by $\cL^{\textrm{EB}}=\{\Lambda^{\textrm{EB}}\}$ the class of entanglement-breaking channels~\footnote{We are here considering maps with the same fixed input/output space, but varying $k$.}. We can then write $\cL^{\textrm{EB}}\subseteq \cL^{\textrm{Sym}_+(k)}$ and $\cL^{\textrm{Sym}_+(k)}\rightarrow \cL^{\textrm{EB}}$ for $k\rightarrow\infty$.

\emph{Mutual information, recoverability, and discord}---The mutual information between  $A$ and $B$ is defined as $I(A:B)_\rho=S(A)_\rho+S(B)_\rho-S(AB)_\rho$, and is a fundamental measure of the total correlations present between $A$ and $B$~\cite{nielsen2010quantum,groismanmutual,wilde2013quantum}. The conditional mutual information can be defined as~\cite{nielsen2010quantum} $I(A:B|C)_\rho=I(A:BC)_\rho-I(A:C)_\rho$, 
i.e., it is equivalent to the decrease of correlations between $A$ and $BC$ due to the loss of system $B$. 
The celebrated strong subadditivity of the von Neumann entropy~\cite{lieb1973proof} is equivalent to 
\beq
\label{eq:strongsubadditivity}
I(A:B|C)_\rho\geq 0.
\eeq
When~\eqref{eq:strongsubadditivity} is satisfied with equality, $\rho_{ABC}$ is said to form a Markov chain: indeed, a strong result by Petz~\cite{petz1986sufficient,petz1988sufficiency,hayden2004structure} ensures that there exist a recovery channel $\cR_{C\rightarrow BC}$ such that $\rho_{ABC}=\cR_{C\rightarrow BC}[\rho_{AC}]$. Fawzi and Renner recently generalized this by proving that, for any tripartite state $\rho_{ABC}$, there always exists a recovery channel $\cR_{C\rightarrow BC}$ such that~\cite{fawzi2014quantum} (see also~\cite{brandao2014quantum})
\beq
\label{eq:FawziRenner}
F(\cR_{C\rightarrow BC}[\rho_{AC}],\rho_{ABC})\geq 2^{-\frac{1}{2}I(A:B|C)_\rho},
\eeq
that is, roughly speaking, the smaller the decrease of correlations 
between $A$ and $BC$ due to the loss of $B$, the better the original $ABC$ state can be recovered from operating just on $C$ alone.

Consider measurement maps $\cM_{B\rightarrow Y}[\cdot] = \sum_y \Tr(M^B_y \cdot) \proj{y}_Y$, where $\{M_y\}$ is a POVM, and the $\ket{y}$s are orthogonal vector states. The discord of $\rho$ between $A$ and $B$ with measurement on $B$ can be defined as~\cite{piani2012problem,seshadreesan2014fidelity}
\beq
\label{eq:discord}
\begin{aligned}
D(A:\underline{B})_{\rho}
&=\min_{\cM_{B\rightarrow Y}}\left(I(A:B)_{\rho_{AB}}-I(A:Y)_{\cM_{B\rightarrow Y}[\rho_{AB}]}\right)\\
&=\min_{V_{B\rightarrow YE}}\left(I(A:YE)_{\rho_{AYE}}-I(A:Y)_{\rho_{AY}}\right)\\
&=\min_{V_{B\rightarrow YE}} I(A:Y|E)_{\rho_{AYE}},
\end{aligned}
\eeq
where in the second and third lines the minimization is over all isometries $V_{B\rightarrow YE}$ that realize measurement maps $\cM_{B\rightarrow Y}$, with $E$ considered as the environment of the dilation~\cite{nielsen2010quantum,seshadreesan2014fidelity}. That is, $E$ is the system that is traced out, or lost, in $\cM_{B\rightarrow Y}[\cdot] = \Tr_E(V_{B\rightarrow YE} \cdot V_{B\rightarrow YE}^\dagger)$,
and $\rho_{AYE} = V_{B\rightarrow YE} \rho_{AB} V_{B\rightarrow YE}^\dagger$. Notice that $I(A:B)_{\rho_{AB}} = I(A:YE)_{\rho_{AYE}}$. It can be proven~\cite{hayashi2006quantum,luo2010decomposition} that the only states with vanishing discord are quantum-classical states of the form~\eqref{eq:QCstate}.

In the case of a (local) measurement, the recovery map (for our intentions, directly to $B$, rather than $YE$) can be assumed to be of the form~\cite{seshadreesan2014fidelity} $\cR_{Y\rightarrow B}[\cdot] = \sum_k \Tr(\proj{y}_Y \cdot ) \sigma^y_B$, 
with $\sigma^y_B$ states, so that the combination of measurement and recovery, $\cR_{Y\rightarrow B}\circ \cM_{B\rightarrow Y}$, is an entanglement-breaking map~\eqref{eq:EB}~\cite{horodecki2003entanglement}. Then, combining \eqref{eq:FawziRenner} and \eqref{eq:discord}, one has~\cite{seshadreesan2014fidelity}
\beq
\label{eq:wildediscord}
\sup_{\Lambda^{\text{EB}}\in\cL^{\text{EB}}}F(\Lambda^{\text{EB}}_B[\rho_{AB}],\rho_{AB}) \geq 2^{-\frac{1}{2} D(A:\underline{B})}.
\eeq
Introducing the surprisal of measurement recoverability~\cite{seshadreesan2014fidelity}   
$D_F(A:\underline{B}):=-\log\sup_{\Lambda^{\text{EB}}\in\cL^{\text{EB}}}F^2(\Lambda^{\text{EB}}_B[\rho_{AB}],\rho_{AB})$, one can cast~\eqref{eq:wildediscord} as $D_F(A:\underline{B}) \leq D(A:\underline{B})$. The surprisal of measurement recoverability quantifies the necessary disturbance introduced by manipulating locally (on $B$) the state $\rho_{AB}$, through measurement and preparation. 
Notice that this can be generalized to any class of maps that correspond to a non-trivial (local) manipulation~(see~\cite{piani2014quantumness}), i.e., one can consider
$D_{F,\cL}(A:\underline{B}):=-\log\sup_{\Lambda\in\cL}F^2(\Lambda_B[\rho_{AB}],\rho_{AB})$, for $\cL$ some class of channels.
With this notation, we can write $D_F(A:\underline{B})=D_{F,\cL^{\text{EB}}}(A:\underline{B})$, where, we recall, $\cL^{\text{EB}}$ indicates the set of entanglement-breaking channels.
Notice that if $\cL^{\textrm{EB}}\subseteq \cL$, it necessarily holds
\beq
D_{F,\cL}(A:\underline{B}) \leq  D_{F,\cL^{\text{EB}}}(A:\underline{B}) \leq D(A:\underline{B}).
\eeq
In particular, we will consider $\cL=\cL^{\textrm{Sym}_+(k)}$. Notice that, in the other direction, Eq.~\eqref{eq:chiribella} implies~(see Appendix)
$\sup_{\Lambda^{\text{EB}}}
F(\rho_{AB},\Lambda^{\text{EB}}_B[\rho_{AB}])
\geq 
\sup_{\Lambda^{\textrm{Sym}_+(k)}}
F(\rho_{AB},\Lambda^{\textrm{Sym}_+(k)}_B[\rho_{AB}]) - \sqrt{(2|B|)/k}$,
so $D_{F,\cL^{\textrm{Sym}_+(k)}}(A:\underline{B})\rightarrow D_{F,\cL^{\text{EB}}}(A:\underline{B})$ for $k\rightarrow \infty$.

\emph{Choi-Jamio{\l}kowski isomorphism and $k$-extendible maps.}---The Choi-Jamio{\l}kowski isomorphism~\cite{choi,jamiolkowski} is a one-to-one correspondence between linear maps $\Lambda_{X\rightarrow Y} $ from $L(\cH_X)$ to $L(\cH_Y)$ and linear operators $W_{XY}$ in $L(\cH_X\otimes\cH_Y)$. It reads
\beq
\label{eq:choi}
J(\Lambda)_{XY} =(\idmap_{X}\otimes\Lambda_{X'\rightarrow Y})[\tilde{\psi}^+_{XX'}],
\eeq
with inverse
\beq
\label{eq:choiinverse}
(J^{-1}(W_{XY}))_{X\rightarrow Y} [\rho_{X}] =\Tr_{X}(W^{\Gamma_{X}}_{XY}\rho_{X}).
\eeq
Here $\tilde{\psi}^+_{XX'}=\proj{\tilde{\psi}^+}_{XX'}$, with the unnormalized maximally entangled state $\ket{\tilde{\psi}^+}_{XX'}=\sum_x \ket{x}_{X}\ket{x}_{X'}$, for $\{\ket{x}\}$ an orthonormal basis of $\cH_X$, and $^{\Gamma_{X}}$ indicates partial transposition on $X$.
The operator $J(\Lambda)$ encodes all the information about the map $\Lambda$. In particular, the linear map $(J^{-1}(W_{XY}))_{X\rightarrow Y}$ defined via~\eqref{eq:choiinverse} is a valid quantum channel from $X$ to $Y$ if and only if $W_{XY}$ is positive semidefinite and $W_X=\Tr_{Y} (W_{XY})=\openone_{X}$. Also, $(J^{-1}(W_{XY}))_{X\rightarrow Y}$ is an entanglement breaking channel if and only if $W_{XY}$ satisfies the additional condition of being proportional to a separable state. Finally, it is easily checked that $J^{-1}(W_{XY})_{X\rightarrow Y}$ is  a $k$-Bose-symmetric extendible channel if and only if, besides satisfying the conditions to be isomorphic to a channel, $W_{XY}$ admits $k$-Bose-symmetric extensions on $Y$.


\emph{A faithful SDP lower bound to quantum discord}---The major obstacle in the computation of the surprisal of measurement recoverability is the fact that it requires an optimization over entanglement breaking channels, i.e., via the Choi-Jamio{\l}kowski isomorphism, over separable states, which cannot be easily parametrized. 

In our case, relaxing the problem, we choose to maximize the fidelity between $\rho=\rho_{AB}$ and $\sigma=(\idmap_A\otimes\Lambda_{B}^{\textrm{Sym}_+(k)})[\rho_{AB}])$, optimizing over $\Lambda_{B}^{\textrm{Sym}_+(k)}\in\cL^{\textrm{Sym}_+(k)}$.
The Choi-Jamio{\l}koski isomorphism allows us to write this as an optimization over positive semidefinite operators $W_{BB'}$ that satisfy $W_B=\openone_B$ and admit $k$-Bose-symmetric extensions. Hence we can write this as an optimization over extended operators $W_{BB^k}$ isomorphic to $k$-Bose-symmetric broadcasting channels. Putting everything together, we find that $\sup_{\Lambda\in\cL^{\textrm{Sym(k)}}}F(\rho_{AB},\Lambda_B[\rho_{AB}])$, from which $D_{F,\cL^{\textrm{Sym}_+(k)}}(A:\underline{B})$ can be derived, corresponds to the solution of the following SDP optimization problem:
\begin{subequations}
  \label{eq:SDPbroadcastfidelity}
  \begin{align}
    {\text{maximize}}\quad
    & \frac{1}{2}(\Tr(X)+\Tr(X^\dagger)) \\
    \text{subject to}\quad
    &  \begin{pmatrix}
    	\rho_{AB} 		& X \\
     	X^\dagger 	& \Tr_{\backslash AB_1}(W^{\Gamma_B}_{BB^k}\rho_{AB})
     \end{pmatrix}
     \geq 0 \\
    & W_{BB^k} \geq 0 \\
    & W_{B} = \openone_B \\  
    & W_{BB^k} = \Pi^+_{B^k} W_{BB^k} \Pi^+_{B^k}. 
  \end{align}
\end{subequations}
We already argued that $D_{F,\cL^{\textrm{Sym}_+(k)}}(A:\underline{B})$ converges to $D_{F,\cL^{\textrm{EB}}}(A:\underline{B})$. To see that it does so monotonically, i.e., that $D_{F,\cL^{\textrm{Sym}(k+1)}}(A:\underline{B})\geq D_{F,\cL^{\textrm{Sym}_+(k)}}(A:\underline{B})$, it is enough to notice that, if $W_{BB^{k+1}}$ is Bose-symmetric on $B^{k+1}$, then $\Tr_{B_{k+1}}(W_{BB^{k+1}})$ is Bose-symmetric on $B^k$.  We also remark again that $D_{F,\cL^{\textrm{Sym}(2)}}(A:\underline{B})$ is already a faithful quantifier of discord, in the sense that, thanks to the no-local-broadcasting theorem, we know it is strictly positive for any state that is not classical on $B$. Finally, thanks to the properties of the fidelity $F$, in particular its
 monotonicity under quantum operations, i.e., $F(\Lambda[\sigma],\Lambda[\rho])\geq F(\sigma,\rho)$~\cite{nielsen2010quantum}, it is immediate to check that each $D_{F,\cL^{\textrm{Sym}_+(k)}}(A:\underline{B})$ is invariant under local unitaries on $B$, and monotonically decreasing under general local operations on $A$~\footnote{A detailed proof for the case of $D_{F,\cL^{\textrm{EB}}}$, which can be immediately adapted to $D_{F,\cL^{\textrm{Sym}_+(k)}}$, is presented in~\cite{seshadreesan2014fidelity}.}. Thus, each $D_{F,\cL^{\textrm{Sym}_+(k)}}$, in particular in the case $k=2$, constitutes in itself a well-behaved measure of the general quantumness of correlations~\cite{brodutch2012criteria,piani2012problem}.
 
 \begin{figure}
\begin{center}
\includegraphics[scale=0.35]{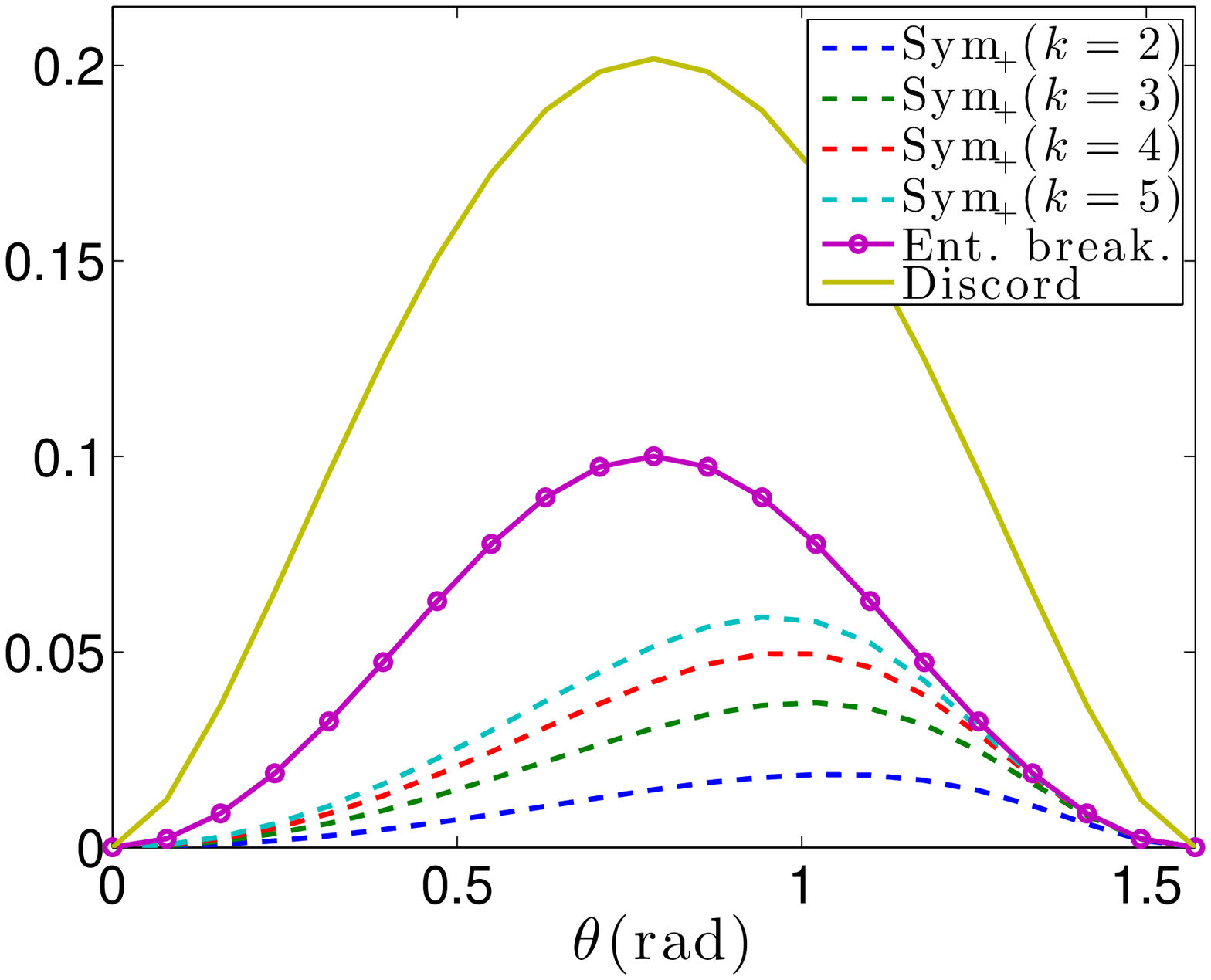}
 \end{center}
 \caption{A hierarchy of lower bounds to quantum discord (colour online). We consider the class of states $\rho_{AB}(\theta)=\frac{1}{2}\proj{0}_A\otimes\proj{\psi_0(\theta)}_B+\frac{1}{2}\proj{1}_A\otimes\proj{\psi_1(\theta)}_B$, with $\ket{\psi_a(\theta)}= \cos(\theta/2)\ket{0}+(-1)^a\sin(\theta/2)\ket{1}$, $a=0,1$, for $\theta\in[0,\pi/2]$.
 From bottom to top, we plot $D_{F,\cL^{\textrm{Sym}_+(k)}}$ for $k=2,3,4,5$ (dashed lines),
  $D_{F,\cL^{\textrm{EB}}}$ (line with circles), as calculated via SDP, 
  and the discord proper $D$ (on $B$) (solid line) as calculated analytically in~\cite{fuchstwo,yao2013quantum}. The state $\rho_{AB}(\theta)$ is classical on $B$ only for $\theta=0,\pi/2$, and any element in the hierarchy verifies this quantitatively. See the main text for definitions.}
 \label{fig:graph}
 \end{figure}

Notice that, if the goal is that of lower-bounding the surprisal of measurement recoverability---and in turn standard discord---rather than just considering a class of physical channels like Bose-symmetric extendible ones, we can impose additional `unphysical' properties that nonetheless make the considered class more closely approximate the class of entanglement-breaking channels. Correspondingly, the SDP optimization~\eqref{eq:SDPbroadcastfidelity} can be modified to include additional constraints, in particular asking for $W_{BB^k}$ to be positive under partial transposition (PPT) in any bipartite cut. In particular, simply by asking that it is PPT with respect to the $B:B^k$ partition, e.g., by adding to~\eqref{eq:SDPbroadcastfidelity} the condition $W_{BB^k}^{\Gamma_B}\geq 0$, we make the corresponding $k$-Bose-extendible channel PPT binding~\cite{PPbinding}, i.e., such that the state $(\idmap_A\otimes\Lambda_B)[\sigma_{AB}]$ is PPT for all $\sigma_{AB}$. This is a non-trivial constraint also for the case $k=1$, and, in the case $|B|=2$, enough to make the channel entanglement breaking~\cite{horodecki1996separability} so that in this case the solution to the SDP provides exactly the surprisal of measurement recoverability. We implemented \eqref{eq:SDPbroadcastfidelity} in MATLAB~\cite{MATLAB}, making use of CVX~\cite{cvx,gb08} and other tools publicly available~\cite{qetlab,cubittcode}. An example of the results is presented in Figure~\ref{fig:graph}.

\emph{Discord, entanglement, and symmetric extensions.}---Our approach, based on an SDP hierarchy dealing with symmetric extensions, is inspired by and very similar to the one used to verify entanglement~\cite{DohertyPRL,DohertyPRA} (see also~\cite{nowakowski} for applications to the extendability of channels). In turn, the fact that fidelity can be expressed as an SDP program, which we exploited here, could also be adopted for the study and quantification of entanglement, providing a hierarchy of SDP programs that allows to calculate the largest fidelity of the given state $\rho_{AB}$ with any state $\sigma^{\textrm{Sym(k)}}_{AB}$ admitting a (Bose-)$k$-symmetric extension on $B$, and converging to the fidelity of separability~\cite{streltsov2010linking}. Our approach points to a illuminating conceptual relation between entanglement and discord, in terms of symmetric extensions and how they are generated: entanglement limits how well a state can be approximated by a state admitting a $k$-symmetric  extension, and only separable states can be perfectly approximated for all $k\geq 2$; on the other hand, discord limits how well a state can be locally transformed into a (Bose-)$k$-symmetric extension of itself, with only discord-free states that can be perfectly locally broadcast, for any $k\geq 2$. Remarkably, while entanglement can be exactly characterized only in the limit $k\rightarrow \infty$, discord can be pinned down already by considering the case $k=2$---this is the content of the no-local-broadcasting theorem. This explains why, while entanglement verification is hard~\cite{gurvits2003classical,ioannou2007computational,gharibian2010}, our hierarchy provides a faithful, reliable, and efficiently computable lower bound to discord already at the lowest level. 

\emph{Conclusions.}---We have introduced a hierarchy of discord-like quantifiers. They are defined in terms of how well a given quantum state $\rho_{AB}$ can be locally broadcast. More precisely, in the lowest non-trivial level of the hierarchy, our quantifier answers the following question: Consider any mapping from $B$ to the  symmetric subspace of two copies $B_1B_2$ of $B$; how well can the resulting $\rho_{AB_1}$ (equivalently, $\rho_{AB_2}$) approximate the original $\rho_{AB}$? In the limit where we consider infinite copies of $B$, instead of just two, the question becomes that of how well the information about $B$ contained in $\rho_{AB}$ can be transmitted (equivalently, stored) in the form of classical information, through a measure, transmit (store), and re-prepare process. Our hierarchy is faithful at all non-trivial levels, i.e., the quantifiers are non-vanishing for states that are not classical on $B$. Each element in the hierarchy corresponds to an SDP optimization problem; hence, it can be reliably and efficiently (in the dimensions of the systems) computed numerically~\cite{DohertyPRL,DohertyPRA}. Furthermore, while each element has a clear physical meaning in itself and satisfies the basic properties to be expected for a meaningful quantifier of the quantumness, it also constitutes a lower bound to the standard quantum discord. Remarkably, in the case in which we are interested in the discord features of a qubit-qudit system, with measurement on the qudit, a tailored SDP program can provide exactly, i.e., up to numerical error, the surprisal of measurement recoverability defined by Seshadreesan and Wilde~\cite{seshadreesan2014fidelity}, and thus  the best possible lower bound to standard quantum discord based on the breakthrough result about quantum Markov chains of Fawzi and Renner~\cite{fawzi2014quantum}. 

\emph{Acknowledgements}---I acknowledge support from NSERC. I would like to thank K. P. Seshadreesan and M. M. Wilde for discussions, and G. Adesso for useful correspondence.


\begin{thebibliography}{80}
\expandafter\ifx\csname natexlab\endcsname\relax\def\natexlab#1{#1}\fi
\expandafter\ifx\csname bibnamefont\endcsname\relax
  \def\bibnamefont#1{#1}\fi
\expandafter\ifx\csname bibfnamefont\endcsname\relax
  \def\bibfnamefont#1{#1}\fi
\expandafter\ifx\csname citenamefont\endcsname\relax
  \def\citenamefont#1{#1}\fi
\expandafter\ifx\csname url\endcsname\relax
  \def\url#1{\texttt{#1}}\fi
\expandafter\ifx\csname urlprefix\endcsname\relax\def\urlprefix{URL }\fi
\providecommand{\bibinfo}[2]{#2}
\providecommand{\eprint}[2][]{\url{#2}}

\bibitem[{\citenamefont{Brunner et~al.}(2014)\citenamefont{Brunner, Cavalcanti,
  Pironio, Scarani, and Wehner}}]{revnonloc}
\bibinfo{author}{\bibfnamefont{N.}~\bibnamefont{Brunner}},
  \bibinfo{author}{\bibfnamefont{D.}~\bibnamefont{Cavalcanti}},
  \bibinfo{author}{\bibfnamefont{S.}~\bibnamefont{Pironio}},
  \bibinfo{author}{\bibfnamefont{V.}~\bibnamefont{Scarani}}, \bibnamefont{and}
  \bibinfo{author}{\bibfnamefont{S.}~\bibnamefont{Wehner}},
  \bibinfo{journal}{Rev. Mod. Phys.} \textbf{\bibinfo{volume}{86}},
  \bibinfo{pages}{419} (\bibinfo{year}{2014}),
  \urlprefix\url{http://link.aps.org/doi/10.1103/RevModPhys.86.419}.

\bibitem[{\citenamefont{Wiseman et~al.}(2007)\citenamefont{Wiseman, Jones, and
  Doherty}}]{wisemanPRL2007}
\bibinfo{author}{\bibfnamefont{H.~M.} \bibnamefont{Wiseman}},
  \bibinfo{author}{\bibfnamefont{S.~J.} \bibnamefont{Jones}}, \bibnamefont{and}
  \bibinfo{author}{\bibfnamefont{A.~C.} \bibnamefont{Doherty}},
  \bibinfo{journal}{Phys. Rev. Lett.} \textbf{\bibinfo{volume}{98}},
  \bibinfo{pages}{140402} (\bibinfo{year}{2007}),
  \urlprefix\url{http://link.aps.org/doi/10.1103/PhysRevLett.98.140402}.

\bibitem[{\citenamefont{Horodecki et~al.}(2009)\citenamefont{Horodecki,
  Horodecki, Horodecki, and Horodecki}}]{revent}
\bibinfo{author}{\bibfnamefont{R.}~\bibnamefont{Horodecki}},
  \bibinfo{author}{\bibfnamefont{P.}~\bibnamefont{Horodecki}},
  \bibinfo{author}{\bibfnamefont{M.}~\bibnamefont{Horodecki}},
  \bibnamefont{and}
  \bibinfo{author}{\bibfnamefont{K.}~\bibnamefont{Horodecki}},
  \bibinfo{journal}{Rev. Mod. Phys.} \textbf{\bibinfo{volume}{81}},
  \bibinfo{pages}{865} (\bibinfo{year}{2009}),
  \urlprefix\url{http://link.aps.org/doi/10.1103/RevModPhys.81.865}.

\bibitem[{\citenamefont{Modi et~al.}(2012)\citenamefont{Modi, Brodutch, Cable,
  Paterek, and Vedral}}]{modi2012classical}
\bibinfo{author}{\bibfnamefont{K.}~\bibnamefont{Modi}},
  \bibinfo{author}{\bibfnamefont{A.}~\bibnamefont{Brodutch}},
  \bibinfo{author}{\bibfnamefont{H.}~\bibnamefont{Cable}},
  \bibinfo{author}{\bibfnamefont{T.}~\bibnamefont{Paterek}}, \bibnamefont{and}
  \bibinfo{author}{\bibfnamefont{V.}~\bibnamefont{Vedral}},
  \bibinfo{journal}{Reviews of Modern Physics} \textbf{\bibinfo{volume}{84}},
  \bibinfo{pages}{1655} (\bibinfo{year}{2012}).

\bibitem[{\citenamefont{Nielsen and Chuang}(2010)}]{nielsen2010quantum}
\bibinfo{author}{\bibfnamefont{M.~A.} \bibnamefont{Nielsen}} \bibnamefont{and}
  \bibinfo{author}{\bibfnamefont{I.~L.} \bibnamefont{Chuang}}
  (\bibinfo{year}{2010}).

\bibitem[{\citenamefont{Gisin et~al.}(2002)\citenamefont{Gisin, Ribordy,
  Tittel, and Zbinden}}]{QKD}
\bibinfo{author}{\bibfnamefont{N.}~\bibnamefont{Gisin}},
  \bibinfo{author}{\bibfnamefont{G.}~\bibnamefont{Ribordy}},
  \bibinfo{author}{\bibfnamefont{W.}~\bibnamefont{Tittel}}, \bibnamefont{and}
  \bibinfo{author}{\bibfnamefont{H.}~\bibnamefont{Zbinden}},
  \bibinfo{journal}{Rev. Mod. Phys.} \textbf{\bibinfo{volume}{74}},
  \bibinfo{pages}{145} (\bibinfo{year}{2002}),
  \urlprefix\url{http://link.aps.org/doi/10.1103/RevModPhys.74.145}.

\bibitem[{\citenamefont{Giovannetti et~al.}(2011)\citenamefont{Giovannetti,
  Lloyd, and Maccone}}]{giovannetti2011advances}
\bibinfo{author}{\bibfnamefont{V.}~\bibnamefont{Giovannetti}},
  \bibinfo{author}{\bibfnamefont{S.}~\bibnamefont{Lloyd}}, \bibnamefont{and}
  \bibinfo{author}{\bibfnamefont{L.}~\bibnamefont{Maccone}},
  \bibinfo{journal}{Nature Photonics} \textbf{\bibinfo{volume}{5}},
  \bibinfo{pages}{222} (\bibinfo{year}{2011}).

\bibitem[{\citenamefont{Ferraro et~al.}(2010)\citenamefont{Ferraro, Aolita,
  Cavalcanti, Cucchietti, and Acin}}]{ferraro2010almost}
\bibinfo{author}{\bibfnamefont{A.}~\bibnamefont{Ferraro}},
  \bibinfo{author}{\bibfnamefont{L.}~\bibnamefont{Aolita}},
  \bibinfo{author}{\bibfnamefont{D.}~\bibnamefont{Cavalcanti}},
  \bibinfo{author}{\bibfnamefont{F.}~\bibnamefont{Cucchietti}},
  \bibnamefont{and} \bibinfo{author}{\bibfnamefont{A.}~\bibnamefont{Acin}},
  \bibinfo{journal}{Physical Review A} \textbf{\bibinfo{volume}{81}},
  \bibinfo{pages}{052318} (\bibinfo{year}{2010}).

\bibitem[{\citenamefont{Luo}(2008)}]{luo2008using}
\bibinfo{author}{\bibfnamefont{S.}~\bibnamefont{Luo}},
  \bibinfo{journal}{Physical Review A} \textbf{\bibinfo{volume}{77}},
  \bibinfo{pages}{022301} (\bibinfo{year}{2008}).

\bibitem[{\citenamefont{Modi et~al.}(2010)\citenamefont{Modi, Paterek, Son,
  Vedral, and Williamson}}]{modi2010unified}
\bibinfo{author}{\bibfnamefont{K.}~\bibnamefont{Modi}},
  \bibinfo{author}{\bibfnamefont{T.}~\bibnamefont{Paterek}},
  \bibinfo{author}{\bibfnamefont{W.}~\bibnamefont{Son}},
  \bibinfo{author}{\bibfnamefont{V.}~\bibnamefont{Vedral}}, \bibnamefont{and}
  \bibinfo{author}{\bibfnamefont{M.}~\bibnamefont{Williamson}},
  \bibinfo{journal}{Physical review letters} \textbf{\bibinfo{volume}{104}},
  \bibinfo{pages}{080501} (\bibinfo{year}{2010}).

\bibitem[{\citenamefont{Daki{\'c} et~al.}(2010)\citenamefont{Daki{\'c}, Vedral,
  and Brukner}}]{dakic2010necessary}
\bibinfo{author}{\bibfnamefont{B.}~\bibnamefont{Daki{\'c}}},
  \bibinfo{author}{\bibfnamefont{V.}~\bibnamefont{Vedral}}, \bibnamefont{and}
  \bibinfo{author}{\bibfnamefont{{\v{C}}.}~\bibnamefont{Brukner}},
  \bibinfo{journal}{Physical review letters} \textbf{\bibinfo{volume}{105}},
  \bibinfo{pages}{190502} (\bibinfo{year}{2010}).

\bibitem[{\citenamefont{Luo and Fu}(2010)}]{luo2010geometric}
\bibinfo{author}{\bibfnamefont{S.}~\bibnamefont{Luo}} \bibnamefont{and}
  \bibinfo{author}{\bibfnamefont{S.}~\bibnamefont{Fu}},
  \bibinfo{journal}{Physical Review A} \textbf{\bibinfo{volume}{82}},
  \bibinfo{pages}{034302} (\bibinfo{year}{2010}).

\bibitem[{\citenamefont{Streltsov et~al.}(2011)\citenamefont{Streltsov,
  Kampermann, and Bru{\ss}}}]{streltsov2011linking}
\bibinfo{author}{\bibfnamefont{A.}~\bibnamefont{Streltsov}},
  \bibinfo{author}{\bibfnamefont{H.}~\bibnamefont{Kampermann}},
  \bibnamefont{and} \bibinfo{author}{\bibfnamefont{D.}~\bibnamefont{Bru{\ss}}},
  \bibinfo{journal}{Physical review letters} \textbf{\bibinfo{volume}{106}},
  \bibinfo{pages}{160401} (\bibinfo{year}{2011}).

\bibitem[{\citenamefont{Piani et~al.}(2011)\citenamefont{Piani, Gharibian,
  Adesso, Calsamiglia, Horodecki, and Winter}}]{piani2011all}
\bibinfo{author}{\bibfnamefont{M.}~\bibnamefont{Piani}},
  \bibinfo{author}{\bibfnamefont{S.}~\bibnamefont{Gharibian}},
  \bibinfo{author}{\bibfnamefont{G.}~\bibnamefont{Adesso}},
  \bibinfo{author}{\bibfnamefont{J.}~\bibnamefont{Calsamiglia}},
  \bibinfo{author}{\bibfnamefont{P.}~\bibnamefont{Horodecki}},
  \bibnamefont{and} \bibinfo{author}{\bibfnamefont{A.}~\bibnamefont{Winter}},
  \bibinfo{journal}{Physical review letters} \textbf{\bibinfo{volume}{106}},
  \bibinfo{pages}{220403} (\bibinfo{year}{2011}).

\bibitem[{\citenamefont{Girolami and Adesso}(2012)}]{girolami2012observable}
\bibinfo{author}{\bibfnamefont{D.}~\bibnamefont{Girolami}} \bibnamefont{and}
  \bibinfo{author}{\bibfnamefont{G.}~\bibnamefont{Adesso}},
  \bibinfo{journal}{Physical review letters} \textbf{\bibinfo{volume}{108}},
  \bibinfo{pages}{150403} (\bibinfo{year}{2012}).

\bibitem[{\citenamefont{Piani and Adesso}(2012)}]{piani2012quantumness}
\bibinfo{author}{\bibfnamefont{M.}~\bibnamefont{Piani}} \bibnamefont{and}
  \bibinfo{author}{\bibfnamefont{G.}~\bibnamefont{Adesso}},
  \bibinfo{journal}{Physical Review A} \textbf{\bibinfo{volume}{85}},
  \bibinfo{pages}{040301} (\bibinfo{year}{2012}).

\bibitem[{\citenamefont{Paula et~al.}(2013)\citenamefont{Paula, de~Oliveira,
  and Sarandy}}]{paula2013geometric}
\bibinfo{author}{\bibfnamefont{F.}~\bibnamefont{Paula}},
  \bibinfo{author}{\bibfnamefont{T.~R.} \bibnamefont{de~Oliveira}},
  \bibnamefont{and} \bibinfo{author}{\bibfnamefont{M.}~\bibnamefont{Sarandy}},
  \bibinfo{journal}{Physical Review A} \textbf{\bibinfo{volume}{87}},
  \bibinfo{pages}{064101} (\bibinfo{year}{2013}).

\bibitem[{\citenamefont{Chang and Luo}(2013)}]{chang2013remedying}
\bibinfo{author}{\bibfnamefont{L.}~\bibnamefont{Chang}} \bibnamefont{and}
  \bibinfo{author}{\bibfnamefont{S.}~\bibnamefont{Luo}},
  \bibinfo{journal}{Physical Review A} \textbf{\bibinfo{volume}{87}},
  \bibinfo{pages}{062303} (\bibinfo{year}{2013}).

\bibitem[{\citenamefont{Girolami et~al.}(2014)\citenamefont{Girolami, Souza,
  Giovannetti, Tufarelli, Filgueiras, Sarthour, Soares-Pinto, Oliveira, and
  Adesso}}]{girolami2014quantum}
\bibinfo{author}{\bibfnamefont{D.}~\bibnamefont{Girolami}},
  \bibinfo{author}{\bibfnamefont{A.~M.} \bibnamefont{Souza}},
  \bibinfo{author}{\bibfnamefont{V.}~\bibnamefont{Giovannetti}},
  \bibinfo{author}{\bibfnamefont{T.}~\bibnamefont{Tufarelli}},
  \bibinfo{author}{\bibfnamefont{J.~G.} \bibnamefont{Filgueiras}},
  \bibinfo{author}{\bibfnamefont{R.~S.} \bibnamefont{Sarthour}},
  \bibinfo{author}{\bibfnamefont{D.~O.} \bibnamefont{Soares-Pinto}},
  \bibinfo{author}{\bibfnamefont{I.~S.} \bibnamefont{Oliveira}},
  \bibnamefont{and} \bibinfo{author}{\bibfnamefont{G.}~\bibnamefont{Adesso}},
  \bibinfo{journal}{Physical Review Letters} \textbf{\bibinfo{volume}{112}},
  \bibinfo{pages}{210401} (\bibinfo{year}{2014}).

\bibitem[{\citenamefont{Piani et~al.}(2014)\citenamefont{Piani, Narasimhachar,
  and Calsamiglia}}]{piani2014quantumness}
\bibinfo{author}{\bibfnamefont{M.}~\bibnamefont{Piani}},
  \bibinfo{author}{\bibfnamefont{V.}~\bibnamefont{Narasimhachar}},
  \bibnamefont{and}
  \bibinfo{author}{\bibfnamefont{J.}~\bibnamefont{Calsamiglia}},
  \bibinfo{journal}{New Journal of Physics} \textbf{\bibinfo{volume}{16}},
  \bibinfo{pages}{113001} (\bibinfo{year}{2014}),
  \urlprefix\url{http://stacks.iop.org/1367-2630/16/i=11/a=113001}.

\bibitem[{\citenamefont{Seshadreesan et~al.}(2014)\citenamefont{Seshadreesan,
  Berta, and Wilde}}]{seshadreesan2014r}
\bibinfo{author}{\bibfnamefont{K.~P.} \bibnamefont{Seshadreesan}},
  \bibinfo{author}{\bibfnamefont{M.}~\bibnamefont{Berta}}, \bibnamefont{and}
  \bibinfo{author}{\bibfnamefont{M.~M.} \bibnamefont{Wilde}},
  \bibinfo{journal}{arXiv preprint arXiv:1410.1443}  (\bibinfo{year}{2014}).

\bibitem[{\citenamefont{Seshadreesan and
  Wilde}(2014)}]{seshadreesan2014fidelity}
\bibinfo{author}{\bibfnamefont{K.~P.} \bibnamefont{Seshadreesan}}
  \bibnamefont{and} \bibinfo{author}{\bibfnamefont{M.~M.} \bibnamefont{Wilde}},
  \bibinfo{journal}{arXiv preprint arXiv:1410.1441}  (\bibinfo{year}{2014}).

\bibitem[{\citenamefont{Ollivier and Zurek}(2001)}]{ollivier2001quantum}
\bibinfo{author}{\bibfnamefont{H.}~\bibnamefont{Ollivier}} \bibnamefont{and}
  \bibinfo{author}{\bibfnamefont{W.~H.} \bibnamefont{Zurek}},
  \bibinfo{journal}{Physical review letters} \textbf{\bibinfo{volume}{88}},
  \bibinfo{pages}{017901} (\bibinfo{year}{2001}).

\bibitem[{\citenamefont{Henderson and Vedral}(2001)}]{henderson2001classical}
\bibinfo{author}{\bibfnamefont{L.}~\bibnamefont{Henderson}} \bibnamefont{and}
  \bibinfo{author}{\bibfnamefont{V.}~\bibnamefont{Vedral}},
  \bibinfo{journal}{Journal of physics A: mathematical and general}
  \textbf{\bibinfo{volume}{34}}, \bibinfo{pages}{6899} (\bibinfo{year}{2001}).

\bibitem[{\citenamefont{Brodutch and Modi}(2012)}]{brodutch2012criteria}
\bibinfo{author}{\bibfnamefont{A.}~\bibnamefont{Brodutch}} \bibnamefont{and}
  \bibinfo{author}{\bibfnamefont{K.}~\bibnamefont{Modi}},
  \bibinfo{journal}{Quantum Information and Computation}
  \textbf{\bibinfo{volume}{12}}, \bibinfo{pages}{0721} (\bibinfo{year}{2012}).

\bibitem[{\citenamefont{Piani}(2012)}]{piani2012problem}
\bibinfo{author}{\bibfnamefont{M.}~\bibnamefont{Piani}},
  \bibinfo{journal}{Physical Review A} \textbf{\bibinfo{volume}{86}},
  \bibinfo{pages}{034101} (\bibinfo{year}{2012}).

\bibitem[{\citenamefont{Piani et~al.}(2008)\citenamefont{Piani, Horodecki, and
  Horodecki}}]{piani2008no}
\bibinfo{author}{\bibfnamefont{M.}~\bibnamefont{Piani}},
  \bibinfo{author}{\bibfnamefont{P.}~\bibnamefont{Horodecki}},
  \bibnamefont{and}
  \bibinfo{author}{\bibfnamefont{R.}~\bibnamefont{Horodecki}},
  \bibinfo{journal}{Physical review letters} \textbf{\bibinfo{volume}{100}},
  \bibinfo{pages}{090502} (\bibinfo{year}{2008}).

\bibitem[{\citenamefont{Luo and Sun}(2010)}]{luo2010decomposition}
\bibinfo{author}{\bibfnamefont{S.}~\bibnamefont{Luo}} \bibnamefont{and}
  \bibinfo{author}{\bibfnamefont{W.}~\bibnamefont{Sun}},
  \bibinfo{journal}{Physical Review A} \textbf{\bibinfo{volume}{82}},
  \bibinfo{pages}{012338} (\bibinfo{year}{2010}).

\bibitem[{\citenamefont{Barnum et~al.}(1996)\citenamefont{Barnum, Caves, Fuchs,
  Jozsa, and Schumacher}}]{barnum1996noncommuting}
\bibinfo{author}{\bibfnamefont{H.}~\bibnamefont{Barnum}},
  \bibinfo{author}{\bibfnamefont{C.~M.} \bibnamefont{Caves}},
  \bibinfo{author}{\bibfnamefont{C.~A.} \bibnamefont{Fuchs}},
  \bibinfo{author}{\bibfnamefont{R.}~\bibnamefont{Jozsa}}, \bibnamefont{and}
  \bibinfo{author}{\bibfnamefont{B.}~\bibnamefont{Schumacher}},
  \bibinfo{journal}{Physical Review Letters} \textbf{\bibinfo{volume}{76}},
  \bibinfo{pages}{2818} (\bibinfo{year}{1996}).

\bibitem[{\citenamefont{Wootters and Zurek}(1982)}]{wootters1982single}
\bibinfo{author}{\bibfnamefont{W.~K.} \bibnamefont{Wootters}} \bibnamefont{and}
  \bibinfo{author}{\bibfnamefont{W.~H.} \bibnamefont{Zurek}},
  \bibinfo{journal}{Nature} \textbf{\bibinfo{volume}{299}},
  \bibinfo{pages}{802} (\bibinfo{year}{1982}).

\bibitem[{\citenamefont{Streltsov and Zurek}(2013)}]{streltsov2013quantum}
\bibinfo{author}{\bibfnamefont{A.}~\bibnamefont{Streltsov}} \bibnamefont{and}
  \bibinfo{author}{\bibfnamefont{W.~H.} \bibnamefont{Zurek}},
  \bibinfo{journal}{Physical review letters} \textbf{\bibinfo{volume}{111}},
  \bibinfo{pages}{040401} (\bibinfo{year}{2013}).

\bibitem[{\citenamefont{Brandao et~al.}(2013)\citenamefont{Brandao, Piani, and
  Horodecki}}]{brandao2013quantum}
\bibinfo{author}{\bibfnamefont{F.~G.} \bibnamefont{Brandao}},
  \bibinfo{author}{\bibfnamefont{M.}~\bibnamefont{Piani}}, \bibnamefont{and}
  \bibinfo{author}{\bibfnamefont{P.}~\bibnamefont{Horodecki}},
  \bibinfo{journal}{arXiv preprint arXiv:1310.8640}  (\bibinfo{year}{2013}).

\bibitem[{\citenamefont{Petz}(1986)}]{petz1986sufficient}
\bibinfo{author}{\bibfnamefont{D.}~\bibnamefont{Petz}},
  \bibinfo{journal}{Communications in mathematical physics}
  \textbf{\bibinfo{volume}{105}}, \bibinfo{pages}{123} (\bibinfo{year}{1986}).

\bibitem[{\citenamefont{Petz}(1988)}]{petz1988sufficiency}
\bibinfo{author}{\bibfnamefont{D.}~\bibnamefont{Petz}}, \bibinfo{journal}{The
  Quarterly Journal of Mathematics} \textbf{\bibinfo{volume}{39}},
  \bibinfo{pages}{97} (\bibinfo{year}{1988}).

\bibitem[{\citenamefont{Hayden et~al.}(2004)\citenamefont{Hayden, Jozsa, Petz,
  and Winter}}]{hayden2004structure}
\bibinfo{author}{\bibfnamefont{P.}~\bibnamefont{Hayden}},
  \bibinfo{author}{\bibfnamefont{R.}~\bibnamefont{Jozsa}},
  \bibinfo{author}{\bibfnamefont{D.}~\bibnamefont{Petz}}, \bibnamefont{and}
  \bibinfo{author}{\bibfnamefont{A.}~\bibnamefont{Winter}},
  \bibinfo{journal}{Communications in mathematical physics}
  \textbf{\bibinfo{volume}{246}}, \bibinfo{pages}{359} (\bibinfo{year}{2004}).

\bibitem[{\citenamefont{Hayashi}(2006)}]{hayashi2006quantum}
\bibinfo{author}{\bibfnamefont{M.}~\bibnamefont{Hayashi}},
  \emph{\bibinfo{title}{Quantum Information}} (\bibinfo{publisher}{Springer},
  \bibinfo{year}{2006}).

\bibitem[{\citenamefont{Bae and Ac{\'i}n}(2006)}]{bae2006asymptotic}
\bibinfo{author}{\bibfnamefont{J.}~\bibnamefont{Bae}} \bibnamefont{and}
  \bibinfo{author}{\bibfnamefont{A.}~\bibnamefont{Ac{\'i}n}},
  \bibinfo{journal}{Physical review letters} \textbf{\bibinfo{volume}{97}},
  \bibinfo{pages}{030402} (\bibinfo{year}{2006}).

\bibitem[{\citenamefont{Chiribella and DÕAriano}(2006)}]{chiribella2006quantum}
\bibinfo{author}{\bibfnamefont{G.}~\bibnamefont{Chiribella}} \bibnamefont{and}
  \bibinfo{author}{\bibfnamefont{G.~M.} \bibnamefont{DÕAriano}},
  \bibinfo{journal}{Physical review letters} \textbf{\bibinfo{volume}{97}},
  \bibinfo{pages}{250503} (\bibinfo{year}{2006}).

\bibitem[{\citenamefont{Chiribella}(2011)}]{chiribella2011quantum}
\bibinfo{author}{\bibfnamefont{G.}~\bibnamefont{Chiribella}}, in
  \emph{\bibinfo{booktitle}{Theory of Quantum Computation, Communication, and
  Cryptography}} (\bibinfo{publisher}{Springer}, \bibinfo{year}{2011}), pp.
  \bibinfo{pages}{9--25}.

\bibitem[{\citenamefont{Fawzi and Renner}(2014)}]{fawzi2014quantum}
\bibinfo{author}{\bibfnamefont{O.}~\bibnamefont{Fawzi}} \bibnamefont{and}
  \bibinfo{author}{\bibfnamefont{R.}~\bibnamefont{Renner}},
  \bibinfo{journal}{arXiv preprint arXiv:1410.0664}  (\bibinfo{year}{2014}).

\bibitem[{\citenamefont{Brandao et~al.}(2014)\citenamefont{Brandao, Harrow,
  Oppenheim, and Strelchuk}}]{brandao2014quantum}
\bibinfo{author}{\bibfnamefont{F.~G.} \bibnamefont{Brandao}},
  \bibinfo{author}{\bibfnamefont{A.~W.} \bibnamefont{Harrow}},
  \bibinfo{author}{\bibfnamefont{J.}~\bibnamefont{Oppenheim}},
  \bibnamefont{and}
  \bibinfo{author}{\bibfnamefont{S.}~\bibnamefont{Strelchuk}},
  \bibinfo{journal}{arXiv preprint arXiv:1411.4921}  (\bibinfo{year}{2014}).

\bibitem[{\citenamefont{Doherty et~al.}(2002)\citenamefont{Doherty, Parrilo,
  and Spedalieri}}]{DohertyPRL}
\bibinfo{author}{\bibfnamefont{A.~C.} \bibnamefont{Doherty}},
  \bibinfo{author}{\bibfnamefont{P.~A.} \bibnamefont{Parrilo}},
  \bibnamefont{and} \bibinfo{author}{\bibfnamefont{F.~M.}
  \bibnamefont{Spedalieri}}, \bibinfo{journal}{Phys. Rev. Lett.}
  \textbf{\bibinfo{volume}{88}}, \bibinfo{pages}{187904}
  (\bibinfo{year}{2002}),
  \urlprefix\url{http://link.aps.org/doi/10.1103/PhysRevLett.88.187904}.

\bibitem[{\citenamefont{Doherty et~al.}(2004)\citenamefont{Doherty, Parrilo,
  and Spedalieri}}]{DohertyPRA}
\bibinfo{author}{\bibfnamefont{A.~C.} \bibnamefont{Doherty}},
  \bibinfo{author}{\bibfnamefont{P.~A.} \bibnamefont{Parrilo}},
  \bibnamefont{and} \bibinfo{author}{\bibfnamefont{F.~M.}
  \bibnamefont{Spedalieri}}, \bibinfo{journal}{Phys. Rev. A}
  \textbf{\bibinfo{volume}{69}}, \bibinfo{pages}{022308}
  (\bibinfo{year}{2004}),
  \urlprefix\url{http://link.aps.org/doi/10.1103/PhysRevA.69.022308}.

\bibitem[{\citenamefont{Doherty}(2014)}]{DohertyBell}
\bibinfo{author}{\bibfnamefont{A.~C.} \bibnamefont{Doherty}},
  \bibinfo{journal}{Journal of Physics A: Mathematical and Theoretical}
  \textbf{\bibinfo{volume}{47}}, \bibinfo{pages}{424004}
  (\bibinfo{year}{2014}),
  \urlprefix\url{http://stacks.iop.org/1751-8121/47/i=42/a=424004}.

\bibitem[{\citenamefont{Boyd and Vandenberghe}(2009)}]{boyd2009convex}
\bibinfo{author}{\bibfnamefont{S.}~\bibnamefont{Boyd}} \bibnamefont{and}
  \bibinfo{author}{\bibfnamefont{L.}~\bibnamefont{Vandenberghe}},
  \emph{\bibinfo{title}{Convex optimization}} (\bibinfo{publisher}{Cambridge
  University Press}, \bibinfo{year}{2009}).

\bibitem[{\citenamefont{Nowakowski and Horodecki}(2009)}]{nowakowski}
\bibinfo{author}{\bibfnamefont{M.~L.} \bibnamefont{Nowakowski}}
  \bibnamefont{and}
  \bibinfo{author}{\bibfnamefont{P.}~\bibnamefont{Horodecki}},
  \bibinfo{journal}{Journal of Physics A: Mathematical and Theoretical}
  \textbf{\bibinfo{volume}{42}}, \bibinfo{pages}{135306}
  (\bibinfo{year}{2009}),
  \urlprefix\url{http://stacks.iop.org/1751-8121/42/i=13/a=135306}.

\bibitem[{\citenamefont{Jain et~al.}(2011)\citenamefont{Jain, Ji, Upadhyay, and
  Watrous}}]{jain2011qip}
\bibinfo{author}{\bibfnamefont{R.}~\bibnamefont{Jain}},
  \bibinfo{author}{\bibfnamefont{Z.}~\bibnamefont{Ji}},
  \bibinfo{author}{\bibfnamefont{S.}~\bibnamefont{Upadhyay}}, \bibnamefont{and}
  \bibinfo{author}{\bibfnamefont{J.}~\bibnamefont{Watrous}},
  \bibinfo{journal}{Journal of the ACM (JACM)} \textbf{\bibinfo{volume}{58}},
  \bibinfo{pages}{30} (\bibinfo{year}{2011}).

\bibitem[{\citenamefont{Kempe et~al.}(2010)\citenamefont{Kempe, Regev, and
  Toner}}]{kempe2010unique}
\bibinfo{author}{\bibfnamefont{J.}~\bibnamefont{Kempe}},
  \bibinfo{author}{\bibfnamefont{O.}~\bibnamefont{Regev}}, \bibnamefont{and}
  \bibinfo{author}{\bibfnamefont{B.}~\bibnamefont{Toner}},
  \bibinfo{journal}{SIAM Journal on Computing} \textbf{\bibinfo{volume}{39}},
  \bibinfo{pages}{3207} (\bibinfo{year}{2010}).

\bibitem[{\citenamefont{Navascu\'es et~al.}(2007)\citenamefont{Navascu\'es,
  Pironio, and Ac{\'i}n}}]{navascuesPRL}
\bibinfo{author}{\bibfnamefont{M.}~\bibnamefont{Navascu\'es}},
  \bibinfo{author}{\bibfnamefont{S.}~\bibnamefont{Pironio}}, \bibnamefont{and}
  \bibinfo{author}{\bibfnamefont{A.}~\bibnamefont{Ac{\'i}n}},
  \bibinfo{journal}{Phys. Rev. Lett.} \textbf{\bibinfo{volume}{98}},
  \bibinfo{pages}{010401} (\bibinfo{year}{2007}),
  \urlprefix\url{http://link.aps.org/doi/10.1103/PhysRevLett.98.010401}.

\bibitem[{\citenamefont{Watrous}(2009)}]{watrous2009semidefinite}
\bibinfo{author}{\bibfnamefont{J.}~\bibnamefont{Watrous}},
  \bibinfo{journal}{Theory of Computing} \textbf{\bibinfo{volume}{5}}
  (\bibinfo{year}{2009}).

\bibitem[{\citenamefont{Johnston and Kribs}(2010)}]{johnston2010family}
\bibinfo{author}{\bibfnamefont{N.}~\bibnamefont{Johnston}} \bibnamefont{and}
  \bibinfo{author}{\bibfnamefont{D.~W.} \bibnamefont{Kribs}},
  \bibinfo{journal}{Journal of Mathematical Physics}
  \textbf{\bibinfo{volume}{51}}, \bibinfo{pages}{082202}
  (\bibinfo{year}{2010}).

\bibitem[{\citenamefont{Eisert et~al.}(2007)\citenamefont{Eisert, Brandao, and
  Audenaert}}]{eisertSDP}
\bibinfo{author}{\bibfnamefont{J.}~\bibnamefont{Eisert}},
  \bibinfo{author}{\bibfnamefont{F.}~\bibnamefont{Brandao}}, \bibnamefont{and}
  \bibinfo{author}{\bibfnamefont{K.}~\bibnamefont{Audenaert}},
  \bibinfo{journal}{New Journal of Physics} \textbf{\bibinfo{volume}{9}},
  \bibinfo{pages}{46} (\bibinfo{year}{2007}),
  \urlprefix\url{http://stacks.iop.org/1367-2630/9/i=3/a=046}.

\bibitem[{\citenamefont{Skrzypczyk et~al.}(2014)\citenamefont{Skrzypczyk,
  Navascu\'es, and Cavalcanti}}]{quantifyingsteering}
\bibinfo{author}{\bibfnamefont{P.}~\bibnamefont{Skrzypczyk}},
  \bibinfo{author}{\bibfnamefont{M.}~\bibnamefont{Navascu\'es}},
  \bibnamefont{and}
  \bibinfo{author}{\bibfnamefont{D.}~\bibnamefont{Cavalcanti}},
  \bibinfo{journal}{Phys. Rev. Lett.} \textbf{\bibinfo{volume}{112}},
  \bibinfo{pages}{180404} (\bibinfo{year}{2014}),
  \urlprefix\url{http://link.aps.org/doi/10.1103/PhysRevLett.112.180404}.

\bibitem[{\citenamefont{Piani and Watrous}(2014)}]{piani2014einstein}
\bibinfo{author}{\bibfnamefont{M.}~\bibnamefont{Piani}} \bibnamefont{and}
  \bibinfo{author}{\bibfnamefont{J.}~\bibnamefont{Watrous}},
  \bibinfo{journal}{arXiv preprint arXiv:1406.0530}  (\bibinfo{year}{2014}).

\bibitem[{\citenamefont{Killoran}(2012)}]{killoran2012entanglement}
\bibinfo{author}{\bibfnamefont{N.}~\bibnamefont{Killoran}}, Ph.D. thesis,
  \bibinfo{school}{University of Waterloo} (\bibinfo{year}{2012}).

\bibitem[{\citenamefont{Watrous}(2012)}]{watrous2012simpler}
\bibinfo{author}{\bibfnamefont{J.}~\bibnamefont{Watrous}},
  \bibinfo{journal}{arXiv preprint arXiv:1207.5726}  (\bibinfo{year}{2012}).

\bibitem[{\citenamefont{Fuchs and Van De~Graaf}(1999)}]{fuchs1999cryptographic}
\bibinfo{author}{\bibfnamefont{C.~A.} \bibnamefont{Fuchs}} \bibnamefont{and}
  \bibinfo{author}{\bibfnamefont{J.}~\bibnamefont{Van De~Graaf}},
  \bibinfo{journal}{Information Theory, IEEE Transactions on}
  \textbf{\bibinfo{volume}{45}}, \bibinfo{pages}{1216} (\bibinfo{year}{1999}).

\bibitem[{\citenamefont{Werner}(1989)}]{werner1989}
\bibinfo{author}{\bibfnamefont{R.~F.} \bibnamefont{Werner}},
  \bibinfo{journal}{Phys. Rev. A} \textbf{\bibinfo{volume}{40}},
  \bibinfo{pages}{4277} (\bibinfo{year}{1989}),
  \urlprefix\url{http://link.aps.org/doi/10.1103/PhysRevA.40.4277}.

\bibitem[{\citenamefont{Fannes et~al.}(1988)\citenamefont{Fannes, Lewis, and
  Verbeure}}]{fannes1988symmetric}
\bibinfo{author}{\bibfnamefont{M.}~\bibnamefont{Fannes}},
  \bibinfo{author}{\bibfnamefont{J.}~\bibnamefont{Lewis}}, \bibnamefont{and}
  \bibinfo{author}{\bibfnamefont{A.}~\bibnamefont{Verbeure}},
  \bibinfo{journal}{Letters in mathematical physics}
  \textbf{\bibinfo{volume}{15}}, \bibinfo{pages}{255} (\bibinfo{year}{1988}).

\bibitem[{\citenamefont{Raggio and Werner}(1989)}]{raggio1989quantum}
\bibinfo{author}{\bibfnamefont{G.}~\bibnamefont{Raggio}} \bibnamefont{and}
  \bibinfo{author}{\bibfnamefont{R.}~\bibnamefont{Werner}},
  \bibinfo{journal}{Helvetica Physica Acta} \textbf{\bibinfo{volume}{62}},
  \bibinfo{pages}{980} (\bibinfo{year}{1989}).

\bibitem[{\citenamefont{Horodecki et~al.}(2003)\citenamefont{Horodecki, Shor,
  and Ruskai}}]{horodecki2003entanglement}
\bibinfo{author}{\bibfnamefont{M.}~\bibnamefont{Horodecki}},
  \bibinfo{author}{\bibfnamefont{P.~W.} \bibnamefont{Shor}}, \bibnamefont{and}
  \bibinfo{author}{\bibfnamefont{M.~B.} \bibnamefont{Ruskai}},
  \bibinfo{journal}{Reviews in Mathematical Physics}
  \textbf{\bibinfo{volume}{15}}, \bibinfo{pages}{629} (\bibinfo{year}{2003}).

\bibitem[{\citenamefont{Groisman et~al.}(2005)\citenamefont{Groisman, Popescu,
  and Winter}}]{groismanmutual}
\bibinfo{author}{\bibfnamefont{B.}~\bibnamefont{Groisman}},
  \bibinfo{author}{\bibfnamefont{S.}~\bibnamefont{Popescu}}, \bibnamefont{and}
  \bibinfo{author}{\bibfnamefont{A.}~\bibnamefont{Winter}},
  \bibinfo{journal}{Phys. Rev. A} \textbf{\bibinfo{volume}{72}},
  \bibinfo{pages}{032317} (\bibinfo{year}{2005}),
  \urlprefix\url{http://link.aps.org/doi/10.1103/PhysRevA.72.032317}.

\bibitem[{\citenamefont{Wilde}(2013)}]{wilde2013quantum}
\bibinfo{author}{\bibfnamefont{M.~M.} \bibnamefont{Wilde}},
  \emph{\bibinfo{title}{Quantum information theory}}
  (\bibinfo{publisher}{Cambridge University Press}, \bibinfo{year}{2013}).

\bibitem[{\citenamefont{Lieb and Ruskai}(1973)}]{lieb1973proof}
\bibinfo{author}{\bibfnamefont{E.~H.} \bibnamefont{Lieb}} \bibnamefont{and}
  \bibinfo{author}{\bibfnamefont{M.~B.} \bibnamefont{Ruskai}},
  \bibinfo{journal}{Journal of Mathematical Physics}
  \textbf{\bibinfo{volume}{14}}, \bibinfo{pages}{1938} (\bibinfo{year}{1973}).

\bibitem[{\citenamefont{Choi}(1975)}]{choi}
\bibinfo{author}{\bibfnamefont{M.-D.} \bibnamefont{Choi}},
  \bibinfo{journal}{Lin. Alg. Appl.} \textbf{\bibinfo{volume}{10}},
  \bibinfo{pages}{285} (\bibinfo{year}{1975}).

\bibitem[{\citenamefont{Jamio{\l}kowski}(1972)}]{jamiolkowski}
\bibinfo{author}{\bibfnamefont{A.}~\bibnamefont{Jamio{\l}kowski}},
  \bibinfo{journal}{Rep. Math. Phys.} \textbf{\bibinfo{volume}{3}},
  \bibinfo{pages}{275} (\bibinfo{year}{1972}).

\bibitem[{\citenamefont{Fuchs}(2002)}]{fuchstwo}
\bibinfo{author}{\bibfnamefont{C.}~\bibnamefont{Fuchs}}, in
  \emph{\bibinfo{booktitle}{Quantum Communication, Computing, and Measurement
  2}}, edited by \bibinfo{editor}{\bibfnamefont{P.}~\bibnamefont{Kumar}},
  \bibinfo{editor}{\bibfnamefont{G.}~\bibnamefont{D'Ariano}},
  \bibnamefont{and} \bibinfo{editor}{\bibfnamefont{O.}~\bibnamefont{Hirota}}
  (\bibinfo{publisher}{Springer US}, \bibinfo{year}{2002}), pp.
  \bibinfo{pages}{11--16}, ISBN \bibinfo{isbn}{978-0-306-46307-5},
  \urlprefix\url{http://dx.doi.org/10.1007/0-306-47097-7_2}.

\bibitem[{\citenamefont{Yao et~al.}(2013)\citenamefont{Yao, Huang, Zou, and
  Han}}]{yao2013quantum}
\bibinfo{author}{\bibfnamefont{Y.}~\bibnamefont{Yao}},
  \bibinfo{author}{\bibfnamefont{J.-Z.} \bibnamefont{Huang}},
  \bibinfo{author}{\bibfnamefont{X.-B.} \bibnamefont{Zou}}, \bibnamefont{and}
  \bibinfo{author}{\bibfnamefont{Z.-F.} \bibnamefont{Han}},
  \bibinfo{journal}{Quantum Information Processing} pp. \bibinfo{pages}{1--12}
  (\bibinfo{year}{2013}).

\bibitem[{\citenamefont{Horodecki et~al.}(2000)\citenamefont{Horodecki,
  Horodecki, and Horodecki}}]{PPbinding}
\bibinfo{author}{\bibfnamefont{P.}~\bibnamefont{Horodecki}},
  \bibinfo{author}{\bibfnamefont{M.}~\bibnamefont{Horodecki}},
  \bibnamefont{and}
  \bibinfo{author}{\bibfnamefont{R.}~\bibnamefont{Horodecki}},
  \bibinfo{journal}{J. Mod. Opt.} \textbf{\bibinfo{volume}{47}},
  \bibinfo{pages}{347} (\bibinfo{year}{2000}).

\bibitem[{\citenamefont{Horodecki et~al.}(1996)\citenamefont{Horodecki,
  Horodecki, and Horodecki}}]{horodecki1996separability}
\bibinfo{author}{\bibfnamefont{M.}~\bibnamefont{Horodecki}},
  \bibinfo{author}{\bibfnamefont{P.}~\bibnamefont{Horodecki}},
  \bibnamefont{and}
  \bibinfo{author}{\bibfnamefont{R.}~\bibnamefont{Horodecki}},
  \bibinfo{journal}{Physics Letters A} \textbf{\bibinfo{volume}{223}},
  \bibinfo{pages}{1} (\bibinfo{year}{1996}).

\bibitem[{\citenamefont{MATLAB}(2014)}]{MATLAB}
\bibinfo{author}{\bibnamefont{MATLAB}}, \emph{\bibinfo{title}{R2014a}}
  (\bibinfo{publisher}{The MathWorks Inc.}, \bibinfo{address}{Natick,
  Massachusetts}, \bibinfo{year}{2014}).

\bibitem[{\citenamefont{Grant and Boyd}(2014)}]{cvx}
\bibinfo{author}{\bibfnamefont{M.}~\bibnamefont{Grant}} \bibnamefont{and}
  \bibinfo{author}{\bibfnamefont{S.}~\bibnamefont{Boyd}},
  \emph{\bibinfo{title}{{CVX}: Matlab software for disciplined convex
  programming, version 2.1}}, \bibinfo{howpublished}{\url{http://cvxr.com/cvx}}
  (\bibinfo{year}{2014}).

\bibitem[{\citenamefont{Grant and Boyd}(2008)}]{gb08}
\bibinfo{author}{\bibfnamefont{M.}~\bibnamefont{Grant}} \bibnamefont{and}
  \bibinfo{author}{\bibfnamefont{S.}~\bibnamefont{Boyd}}, in
  \emph{\bibinfo{booktitle}{Recent Advances in Learning and Control}}, edited
  by \bibinfo{editor}{\bibfnamefont{V.}~\bibnamefont{Blondel}},
  \bibinfo{editor}{\bibfnamefont{S.}~\bibnamefont{Boyd}}, \bibnamefont{and}
  \bibinfo{editor}{\bibfnamefont{H.}~\bibnamefont{Kimura}}
  (\bibinfo{publisher}{Springer-Verlag Limited}, \bibinfo{year}{2008}), Lecture
  Notes in Control and Information Sciences, pp. \bibinfo{pages}{95--110},
  \bibinfo{note}{\url{http://stanford.edu/~boyd/graph_dcp.html}}.

\bibitem[{\citenamefont{Johnston}(2015)}]{qetlab}
\bibinfo{author}{\bibfnamefont{N.}~\bibnamefont{Johnston}},
  \emph{\bibinfo{title}{{QETLAB}: A {MATLAB} toolbox for quantum entanglement,
  version 0.7}}, \bibinfo{howpublished}{\url{http://qetlab.com}}
  (\bibinfo{year}{2015}).

\bibitem[{\citenamefont{Cubitt}(2015)}]{cubittcode}
\bibinfo{author}{\bibfnamefont{T.}~\bibnamefont{Cubitt}},
  \emph{\bibinfo{title}{Maths code}},
  \bibinfo{howpublished}{\url{http://www.dr-qubit.org/matlab.php}}
  (\bibinfo{year}{2015}).

\bibitem[{\citenamefont{Streltsov et~al.}(2010)\citenamefont{Streltsov,
  Kampermann, and Bru{\ss}}}]{streltsov2010linking}
\bibinfo{author}{\bibfnamefont{A.}~\bibnamefont{Streltsov}},
  \bibinfo{author}{\bibfnamefont{H.}~\bibnamefont{Kampermann}},
  \bibnamefont{and} \bibinfo{author}{\bibfnamefont{D.}~\bibnamefont{Bru{\ss}}},
  \bibinfo{journal}{New Journal of Physics} \textbf{\bibinfo{volume}{12}},
  \bibinfo{pages}{123004} (\bibinfo{year}{2010}).

\bibitem[{\citenamefont{Gurvits}(2003)}]{gurvits2003classical}
\bibinfo{author}{\bibfnamefont{L.}~\bibnamefont{Gurvits}}, in
  \emph{\bibinfo{booktitle}{Proceedings of the thirty-fifth annual ACM
  symposium on Theory of computing}} (\bibinfo{organization}{ACM},
  \bibinfo{year}{2003}), pp. \bibinfo{pages}{10--19}.

\bibitem[{\citenamefont{Ioannou}(2007)}]{ioannou2007computational}
\bibinfo{author}{\bibfnamefont{L.~M.} \bibnamefont{Ioannou}},
  \bibinfo{journal}{Quantum Information \& Computation}
  \textbf{\bibinfo{volume}{7}}, \bibinfo{pages}{335} (\bibinfo{year}{2007}).

\bibitem[{gha(2010)}]{gharibian2010}
\bibinfo{journal}{Quantum Information and Computation}
  \textbf{\bibinfo{volume}{10}}, \bibinfo{pages}{343} (\bibinfo{year}{2010}).

\bibitem[{\citenamefont{Uhlmann}(1976)}]{uhlmann1976transition}
\bibinfo{author}{\bibfnamefont{A.}~\bibnamefont{Uhlmann}},
  \bibinfo{journal}{Reports on Mathematical Physics}
  \textbf{\bibinfo{volume}{9}}, \bibinfo{pages}{273} (\bibinfo{year}{1976}).

\end{thebibliography}

\appendix

\begin{lemma}
\label{lem:quasitriangle}
Consider any three mixed states $\rho$, $\sigma$, and $\tau$. It holds,
\beq
\begin{aligned}
|F(\rho,\sigma)-F(\tau,\sigma)|&\leq \sqrt{2}\sqrt{1-F(\tau,\rho)}\\
				&\leq \sqrt{2}\sqrt{\Delta(\tau,\rho)}
\end{aligned}
\eeq
\end{lemma}
\begin{proof}
Fix an arbitrary purification $\ket{\psi_\rho}$, and choose purifications $\ket{\psi_\sigma}$ and $\ket{\psi_\sigma}$ such that $\braket{\psi_\rho}{\psi_\sigma}=F(\rho,\sigma)$ and $\braket{\psi_\rho}{\psi_\sigma}=F(\rho,\sigma)$.  This is always possible because of Uhlmann's theorem~\cite{uhlmann1976transition,nielsen2010quantum} and by choosing properly phases.

Then,
\begin{subequations}
  	\begin{align*}
		F(\rho,\sigma)	&\notag=|\braket{\psi_\rho}{\psi_\sigma}|\\
					&\notag=|\left(\left(\bra{\psi_\rho}-\bra{\psi_{\tau}}\right)+\bra{\psi_\tau}\right) \ket{\psi_\sigma}|\\
					&\leq |\braket{\psi_\tau}{\psi_\sigma}| + |\left(\bra{\psi_\rho}-\bra{\psi_{\tau}}\right) \ket{\psi_\sigma}|\\
					&\leq F(\tau,\sigma)+\sqrt{\left(\bra{\psi_\rho}-\bra{\psi_{\tau}}\right)\left(\ket{\psi_\rho}-\ket{\psi_{\tau}}\right)}\\
					&\notag=F(\tau,\sigma)+\sqrt{2-\braket{\psi_\tau}{\psi_\rho}-\braket{\psi_\rho}{\psi_\tau}}\\
					&\notag=F(\tau,\sigma)+\sqrt{2}\sqrt{1-F(\tau,\rho)}\\
					&\leq F(\tau,\sigma)+\sqrt{2}\sqrt{\Delta(\tau,\rho)}.
	\end{align*}
\end{subequations}
The first inequality is just the triangle inequality for the absolute value. The second inequality is due to the fact that the fidelity between two states is is the maximum overlap of any two purifications of the states~\cite{uhlmann1976transition,nielsen2010quantum}. The last inequality is due to the standard relation $1-F(\tau,\rho)\leq \Delta(\tau,\rho)$~\cite{fuchs1999cryptographic,nielsen2010quantum}.
\end{proof}

\begin{theorem}
It holds
\begin{multline*}
\sup_{\Lambda^{\text{EB}}} F(\rho_{AB},\Lambda^{\text{EB}}_B[\rho_{AB}])\\
\geq 
\sup_{\Lambda^{\textrm{Sym}_+(k)}}
F(\rho_{AB},\Lambda^{\textrm{Sym}_+(k)}_B[\rho_{AB}]) - \sqrt{\frac{2|B|}{k}}.
\end{multline*}
\end{theorem}
\begin{proof}
Eq. \eqref{eq:chiribella} implies that, for any $\Lambda^{\textrm{Sym}_+(k)}$, there is $\Lambda^{\text{EB}}$ such that, for any $\rho_{AB}$
\[
\Delta\left(\Lambda^{\textrm{Sym}_+(k)}_B[\rho_{AB}],\Lambda^{\text{EB}}_B[\rho_{AB}]\right)\leq \frac{|B|}{k}.
\]
Thus, using Lemma~\ref{lem:quasitriangle}, we obtain
\begin{multline*}
F(\rho_{AB},\Lambda^{\text{EB}}_B[\rho_{AB}])\\
\begin{aligned}
&\geq F(\rho_{AB},\Lambda^{\textrm{Sym}_+(k)}_B[\rho_{AB}])\\
&\quad - \sqrt{2}\sqrt{\Delta\left(\Lambda^{\textrm{Sym}_+(k)}_B[\rho_{AB}],\Lambda^{\text{EB}}_B[\rho_{AB}]\right)}\\
&\geq F(\rho_{AB},\Lambda^{\textrm{Sym}_+(k)}_B[\rho_{AB}]) - \sqrt{\frac{2|B|}{k}}.
\end{aligned}
\end{multline*}
Since this is valid for any $\Lambda^{\textrm{Sym}_+(k)}$, we can take the supremum on both sides over channels in the respective classes.
\end{proof}

\end{document}